\documentclass{llncs}

\usepackage{amsmath, amssymb}
\usepackage{enumerate}
\usepackage{algorithm}
\usepackage{algorithmic}

\title{ %\fbox{Short Paper}\\ \ \\
The 1.375 Approximation Algorithm for Sorting by 
Transpositions Can Run in $O(n\log n)$ Time}

\author{Jesun Sahariar Firoz$^{1,4}$ \and Masud Hasan$^{1,3}$ \and Ashik Zinnat 
Khan$^{1,4}$ \and M. Sohel Rahman$^{1,2,3}$}

\institute{Department of Computer Science and Engineering\\
Bangladesh Unibersity of Engineering and Technology\\
Dhaka-1000, Bangladesh
\and Department of Computer Science, King's College London, UK
\and \email{\{masudhasan,msrahman\}@cse.buet.ac.bd}
\and \email{\{jesunsahariar,ashrik\}@gmail.com}}

\begin{document}
\maketitle

\iffalse
************************ Bioinformatic off *****************

\documentclass{bioinfo}
\copyrightyear{2005}
\pubyear{2005}

\usepackage{amsmath, amssymb}
%\usepackage{enumerate}
%\usepackage{enumerate}
\usepackage{algorithm}
\usepackage{algorithmic}
\usepackage{graphicx}
%\usepackage{textcomp}
\usepackage{times}

\newtheorem{theorem}{Theorem}
\newtheorem{lemma}[theorem]{Lemma}
\newtheorem{proposition}[theorem]{Proposition}
\newtheorem{corollary}[theorem]{Corollary}
\newtheorem{claim}{Claim}

\begin{document}

\firstpage{1}

\title[Sorting by Transpositions in $O(n\log n)$ Time]{The 1.375 Approximation Algorithm for Sorting by Transpositions Can Run in $O(n\log n)$ Time}
\author[Firoz, Hasan, Khan and Rahman]{Jesun Sahariar Firoz\,$^{1}$, Masud Hasan\,$^{1}$, Ashik Zinnat
Khan\,$^{1}$ and M. Sohel Rahman\,$^{1,2,}$\footnote{to whom correspondence should be addressed}}
\address{$^{1}$Department of CSE, BUET, Dhaka-1000, Bangladesh\\
$^{2}$Department of Computer Science, King's College London, London, UK}

\history{Received on XXXXX; revised on XXXXX; accepted on XXXXX}

\editor{Associate Editor: XXXXXXX}

\maketitle

*******************Bioinformatic off *****************************
\fi

\begin{abstract}

%\section{Motivation:}
Sorting a Permutation by Transpositions (SPbT) is an important problem in Bioinformtics. 
In this paper, we improve the running time of the best known approximation algorithm for SPbT.
%\section{Results:}
We use the permutation tree data structure of Feng and Zhu and improve the running time of the 
1.375 Approximation Algorithm for SPbT of Elias and Hartman (EH algorithm) to $O(n\log n)$. 
The previous running time of EH algorithm was $O(n^2)$.
%\section{Availability:}
%Text  Text Text Text Text Text Text Text Text Text  Text Text Text Text Text Text Text Text Text  Text Text Text Text Text Text Text Text Text  Text
%
%In this paper, we improve the running time of the existing best known
%$\mathrm{1.375}-$approximation algorithm for sorting by transpositions with the help of the
%permutation tree data structure. In particular, we improve the running time from
%$O(n^2)$ to $O(n\log n)$.
%
%
%
%\section{Contact:} \href{sohel@dcs.kcl.ac.uk}{sohel@dcs.kcl.ac.uk}
\end{abstract}
%\vspace{-.2in}

\section{Introduction}\label{intro}
Transposition is an important genome rearrangement operation and 
Sorting a Permutation by Transpositions (SPbT) is an important problem in Bioinformtics. 
In the transposition operation, a segment is cut out of
the permutation and pasted in a different location. SPbT was first studied by Bafna and Pevzner \cite{15}, who discussed
the first $\mathrm{1.5}-$approximation algorithm which had quadratic running time. 
Eriksson et al. \cite{33} gave an algorithm that sorts any given permutation of
$n$ elements by at most $\frac{2}{3}{n}$ transpositions. Later, Hartman and Shamir 
used the concept of simplified breakpoint graph to design another $\mathrm{1.5}-$approximation 
algorithm with $O(n^2)$ running time~\cite{34}. They further used the splay tree to implement this simplified
algorithm and thereby reducing the time complexity to ${O}(n^{\frac{3}{2}}\sqrt
{\log n}$)~\cite{34}. Finally, Elias and Hartman presented an $\mathrm{1.375}-$approximation 
algorithm in \cite{36}, which is the best known approximation algorithm for SPbT in the literature so far. 
The running time of that algorithm~\cite{36} however is $O(n^2)$. Very recently, in \cite{31}, Feng and Zhu 
improved the running time of the $\mathrm{1.5}-$approximation algorithm of \cite{34} to
$O(n\log n)$ by introducing and using a new data structure named the \emph{permutation tree}. 
In this paper, with the help of the permutation tree data structure we improve the running 
time of the 1.375 Approximation Algorithm for SPbT of \cite{36} to $O(n\log n)$.
%\vspace{-.2in}

\section {Preliminaries}\label{pre}
A \textit{transposition} $\tau \equiv trans(i,j,k)$  on  $\pi = (\pi_0 . . . \pi_{n-1})$ 
is an exchange of two disjoint contiguous segments \textit{X} = $\pi_i, . . . , \pi_{j-1}
$ and \textit{Y} =$ \pi_j, . . . ,\pi_{k-1}$. Given a permutation $\pi$, the SPbT asks to 
find a sequence of transpositions to transform $\pi$ into the identity permutation such that 
the number of transpositions $t$ is minimized. The \textit{transposition distance} of a permutation $\pi$, denoted by $d(\pi)$,
is the smallest possible value of $t$. The breakpoint
graph $G(\pi)$ \cite{15} is an edge-colored graph on $2n$ vertices $\{l_0,r_0,l_1,r_1,
\ldots , l_{n-1},r_{n-1}\}.$ For every $0 \leq i \leq {n- 1}$, $r_i$ and
$l_{i+1}$ are connected by a grey edge, and for every $\pi_i$, $l_{\pi_i}$ and
$r_{\pi_{i-1}}$ are connected by a black edge, denoted by $b_i$. 
The breakpoint graph uniquely decomposes into c($\pi$) cycles. 
A \emph{$k-$cycle} has $k$ black edges; if $k$ is odd (resp. even), the cycle is \emph{odd} 
(resp. \emph{even}). Further, if $k<3$, it is \emph{short} and else, \emph{long}.
The number of odd cycles is denoted by $c_{odd}(\pi)$, and we define $\Delta c_{odd}(\pi,\tau) 
= c_{odd}(\tau $.$\pi) - c_{odd}(\pi)$, where $\tau.\pi$ denotes the result after $\tau$ is applied. A transposition
$\tau$ is a $k-$move if $\Delta c_{odd}(\pi,\tau) = k$. A cycle is called
\textit{oriented} if there is a 2-move that is applied on three of its black
edges; otherwise, it is \textit{unoriented}. If $G(\pi)$ contains only short cycles, 
then, both $\pi$ and $G(\pi)$ are called \textit{simple}. 
A permutation $\pi$ is \textit{2-permutation} (resp. \textit{3-permutation})
if $G(\pi)$ is contains only 2-cycles (resp. 3-cycles).
Permutations can be made simple by
inserting new elements into the permutations and thereby splitting the long
cycles~\cite{18}.

Two pairs of black edges $(a,b)$ and $(c,d)$ are said to \textit{intersect} if
their edges occur in alternated order in the breakpoint graph, i.e., in order
$a,c,b,d$. Cycles $C$ and $D$ intersect if there is a pair of black edges in $C$
that intersects with a pair of black edges in $D$. A \textit{configuration} of
cycles is a subgraph of the breakpoint graph that is induced by one or more
cycles. Configuration $A$ is \textit{connected} if for any two cycles $c_1$ and
$c_k$ of $A$ there are cycles $c_2, \ldots, c_{k-1} \in A $ such that, for each
$i \in [1,k- 1], c_i$ intersects with $c_{i+1}$. A \textit{component} is a
maximal connected configuration in a breakpoint graph.
The \textit{size} of a configuration or a component is the number of cycles it
contains.
A configuration (similarly, a component) is said to be unoriented if all of its
cycles are unoriented.
A configuration (similarly, a component) is \textit{small} if its size is at
most $8$;
otherwise it is \textit{big}.
Small components that do not have an $\frac{11}{8}$-sequence are called
\textit{bad small components}~\cite{36}.

In a configuration, an \textit{open gate} is a pair of black edges of a 2-cycle
or an unoriented 3-cycle that does not intersect with another cycle of that
configuration.
A configuration not containing open gates is referred to as a full
configuration.
%The following is an important lemma related to the concept of open gate.
%\begin{lemma}\label{lem-5}
%(Bafna and Pevzner \cite{15}.) Every open gate intersects with some other cycle
%in the breakpoint graph.\qed
%\end{lemma}
An $(x, y)-$sequence of transpositions on a simple permutation (for $x\geq y$)
is a sequence of $x$ number of transpositions, such that at least $y$ of them are 2-moves
and that leaves a simple permutation at the end. %For example, a 0-move followed
%by two consecutive 2-moves (which is sometimes referred to as a (0, 2,
%2)-sequence in the literature \cite{16,19}) is a (3, 2)-sequence. An
%$\frac{a}{b}$-sequence is an $(x,y)-$sequence such that $x\leq a$ and
%$\frac{x}{y} \leq \frac{a}{b}$. The algorithm discussed in \cite{36} refers to
%$\frac{11}{8}$-sequences. A configuration (or component or permutation) has an
%$(x,y)$ (or $\frac{a}{b}$) sequence if it is possible to apply such a sequence
%on its cycle. The following lemma is important.

%\begin{lemma}\label{lem-6}
%(Christie \cite{16} and Hartman \cite{34}.) For every permutation (except for
%the identity permutation) there exists either a 2-move or a (3,
%2)-sequence.\qed
%\end{lemma}

%The \emph{transposition diameter}, $TD(n)$, of the symmetric group is the
%maximum value of $d(\pi)$ taken over all permutations of $n$ elements.\\

%\section{Permutation Tree}
A \textit{permutation tree} \cite{31} is firstly a balanced binary tree $T$ with
root $r$, where each internal node of $T$ has two children. 
The left and right children of an internal node $t$ are denoted by $L(t)$ and $R(t)$,
respectively. The height of $t$ is denoted by $H(t)$; a leaf node has height zero. 
%The height of an
%internal node is defined to be $H(t) = max\{H(L(t)),H(R(t))\}+1$.
%Since the tree is balanced, for any node $t$ of $T$, we have $|H(L(t)) -
%H(R(t))| \leq 1$. The height of $T$ is defined to be the height of the root
%$H(T) = H(r)$. 
Secondly, a permutation tree must correspond to a permutation.
The permutation tree corresponding to $\pi$ has $n$
leaf nodes, labeled by $\pi_1,\pi_2, \ldots, \pi_n$ respectively. 
Each node corresponds to an interval of $\pi$ and has a \emph{value} equal to the maximum number in
the interval. The interval corresponding to an internal node 
$t$ is be the concatenation of the two intervals corresponding to $L(t)$ and
$R(t)$. The height of the permutation tree of $\pi$ is bounded by $O(\log |\pi|)$. 
A permutation tree (Build operation) can be built in $O(|\pi|)$ time. 
Suppose, $T, t_1$ and $t_2$ correspond to $(\pi_1\pi_2$ $\ldots\allowbreak\pi_{m-1}\pi_m\pi_{m+1} 
\ldots \pi_n), ~(\pi_1\pi_2 \ldots \pi_m)$ and $(\pi_{m+1}\pi_{m+2}\allowbreak\ldots \pi_n)$, respectively. 
Then, $Join(t_1,t_2)$ returns $T$ in $O(H(t_1)-H(t_2))$ time and $Split(T,m)$ returns $t_l$ and $t_2$ in $O(\log n)$ time.

%
%
%There are  three operations for a permutation tree \cite{31}. They are
%\emph{Build}, which builds a permutation tree corresponding to a given
%permutation, \emph{Join}, which joins two trees into one, and \emph{Split},
%which splits one tree into two. The following theorems report the time
%complexity of these three operations.

%\begin{theorem}\label{thm-8}
%\rm{\cite{31}} For a permutation $\pi =(\pi_1 \pi_2 \ldots \pi_n)$, the time
%complexity of the Build operation is $O(n)$.
%\end{theorem}

%\begin{theorem}\label{thm-9}
%\rm{\cite{31}} If $t_1$ corresponds to $(\pi_1\pi_2 \ldots \pi_m)$, and $t_2$
%corresponds to $(\pi_{m+1}\pi_{m+2}\allowbreak\ldots \pi_n)$, then $Join(t_1,
%t_2)$ returns a permutation tree corresponding to $(\pi_1\pi_2 \ldots
%\pi_m\pi_{m+1}\allowbreak\pi_{m+2} \ldots\pi_n)$. The time complexity of
%$Join(t_1, t_2$) is $O(H(t_1)-H(t_2))$.
%\end{theorem}
%
%\begin{theorem}\label{thm-11}
%\rm{\cite{31}} Suppose $T$ is a permutation tree corresponding to
%$\rho=(\pi_1\pi_2$ $\ldots\allowbreak\pi_{m-1}\pi_m\pi_{m+1} \ldots \pi_n)$.
%The algorithm $Split(T,m)$ always returns $t_l$ corresponding to
%$\rho_l=((\pi_1\pi_2 \ldots\pi_{m-1})$ and $\rho_r=((\pi_m \ldots \pi_n) $. The
%time complexity of
%$Split(T,m)$ is $O(\log n)$.
%\end{theorem}
%\vspace{-.2in}

\section{Faster Running Time for Elias and Hartman's Algorithm}\label{algo}
%\section{Elias and Hartman's $\mathrm{O}$(\textit{$n^2$}) Algorithm}
%The definitions and notations used here mostly follow from \cite{31,36}. 
The 1.375$-$approximation algorithm for SPbT of Elias and Hartman 
(referred to as the EH algorithm henceforth) is presented in Algorithm~\ref{Alg-EliasHartman}.

%\vspace{-.2in}
\begin{algorithm}[h!]
%\scriptsize
\caption{\label{Alg-EliasHartman}EH Algorithm}
\begin{algorithmic}[1]
\STATE \label{Step1}Transform permutation $\pi$ into a simple permutation
$\hat{\pi}$ .
\STATE \label{Step2} Check if there is a $(\mathrm{2,2})$-sequence. If so, apply it.
\STATE \label{Step3}While G($\hat{\pi}$) contains a 2-cycle, apply a 2-move.
\STATE \label{Step4}$\hat{\pi}$ is a 3-permutation. Mark all 3-cycles in G($\hat{\pi}$).

\WHILE{G($\hat{\pi}$) contains a marked 3-cycle $\mathrm{C}$\label{Step5}}
 \IF {$\mathrm{C}$ is oriented}
    \STATE \label{Step5a}apply a 2-move on it.
 \ELSE
    \STATE \label{Step5bStart}Try to sufficiently extend $\mathrm{C}$ eight times
    \IF{sufficient configuration has been achieved}
        \STATE apply an $\dfrac{11}{8}$-sequence.
    \ELSE
        \STATE it must be a small component. If an $\dfrac{11}{8}$-sequence is
still possible apply it.
        \IF{Applying a $\dfrac{11}{8}$-sequence is not possible}
            \STATE \label{Step5bEnd}This must be a bad small component. Unmark
all cycles of the component.
        \ENDIF
    \ENDIF
 \ENDIF
\ENDWHILE

\STATE\label{Step6} Now, G($\hat{\pi}$) contains only bad small components.While
G($\hat{\pi}$) contains \\ atleast 8 cycles,apply an $\dfrac{11}{8}$-sequence.
\STATE\label{Step7}  While G($\hat{\pi}$) contains a 3-cycle,apply a
(3,2)-sequence.
\STATE\label{Step8}  Mimic the sorting of $\pi$ using the sorting of
$\hat{\pi}$.

\end{algorithmic}

\end{algorithm}

%\vspace{-.2in}
Now, to achieve our goal we need to be able to use the permutation tree for applying $(x,y)-$sequence and $k-$move.
%In fact, applying a $(x,y)-$sequence and   a $k-$move to a permutation is equivalent
%to applying $x\ge 1$ number of transpositions or a transposition  respectively.
Additionally, given a pair of black edges
we can find, with the help of a permutation tree, another pair of black edges
such that these two pairs intersect.
Feng and Zhu \cite{31} used the following lemma to find such a pair of black edges.
%\vspace{-.075in}
\begin{lemma}\label{lem-12}
(\cite{15})
Let $b_i$ and $b_j$ are two black edges in an unoriented cycle $C$
such that $i<j$.
Let $\pi_k = \max_{i < m \le j} \pi_m$ and $\pi_l = \pi_k + 1$.
Then the black edges $b_k$ and $b_{l-1} $ belong to the same cycle
and the pair $\langle b_k,b_{l-1}\rangle $ intersects the pair $\langle b_i , b_j \rangle$.\qed
\end{lemma}
%\vspace{-.075in}
Feng and Zhu suggested that a permutation tree can be used for query and transposition as follows. Assume that the permutation tree $T$ corresponding to a simple permutation $\pi=(\pi_1\ldots \pi_n) $ has been constructed by procedure $Build$. Now, Procedure $Query(\pi, i, j)$ finds a pair of black edges intersecting the pair $\langle b_i , b_j \rangle$ and Procedure $Transposition(\pi, i, j, k)$, applies a transposition
$trans(i,j,k)$ on $\pi$. These two procedures can be implemented as follows.
%\vspace{-.15in}
\begin{enumerate}

\item \textbf{$Query(\pi, i, j)$:} Split $T$ into three permutation trees, $t_1,t_2$ and $t_3$, corresponding to, respectively, $(\pi_1,\allowbreak\ldots, \pi_i), (\pi_i+1, \ldots ,\pi_j)$ and $(\pi_j+1, \ldots
,\pi_n)$. Clearly this can be done in $O(\log n)$ time by two splitting 
operations of $T$. The value of the root of $t_2$ is the largest element
(say, $\pi_k$) in the interval $[\pi_i+1\ldots\pi_j ]$. Assume that $\pi_l =\pi_k+1$. By
Lemma \ref{lem-12}, pair $\langle b_k,b_{l-1}\rangle$ intersects pair $\langle
b_i , b_j\rangle $. 
%By Theorems~\ref{thm-9} and~\ref{thm-11}, $Query(\pi, i, j)$
%takes $O(\log n)$ time.

\item \textbf{$Transposition(\pi, i, j, k)$:} Split $T$ into four permutation trees $t_1, t_2, t_3$ and $t_4$, corresponding to, respectively $(\pi_1, \ldots , \pi_{i-1}),~(\pi_i, \ldots ,\pi_{j-1}),~(\pi_j,\allowbreak\ldots,
\pi_k-1)$ and $(\pi_k, \ldots , \pi_n)$. Then, join
the four trees by executing $Join(Join(Join(t_1, t_3),\allowbreak t_2), t_4)$. Clearly, 
adjusting the permutation tree $T$ can be done by three splitting and three
joining operations spending $O(\log n)$ time.
\end{enumerate}
%\vspace{-.15in}
%
%\begin{lemma}\label{lem-13}
%\rm{\cite{31}} The procedures $Query$ and $Transposition$ each can be completed
%in $O(\log n)$ time.
%\end{lemma}

In the rest of this section we state and prove a number of lemmas concerning the running time of different steps of the
the EH algorithm, achieving an $O(n\log n)$ running time for the algorithm in
the sequel.
%\vspace{-.10in}
\begin{lemma}\label{lem-14}
Step~\ref{Step1} of the EH algorithm can be implemented in $O(n)$ time.
\end{lemma}

\begin{proof}
%\vspace{-.10in}
A permutation $\pi$ is made simple by $({g,b})$-splits acting on the breakpoint
graph ${G}(\pi)$. A $({g,b})$-split for ${G}(\pi)$ splits one cycle into two
shorter ones. Equivalently, this operation inserts a new element into
$\pi$~\cite{12}. A breakpoint graph ${G}(\pi)$ can be transformed into
$G(\hat{\pi}) $ containing only 1-cycles, 2-cycles, and 3-cycles by a series of
$({g,b})$-splits~\cite{34}, that is, the permutation corresponding to
$G(\hat{\pi})$ beomces simple. This can be done by scanning the permutation linearly
and inserting a new element when necessary. Thus Step~\ref{Step1} can be
implemented in $O(n)$ time.
\qed
\end{proof}

%Next we consider Step~\ref{Step2} of Algorithm~\ref{Alg-EliasHartman}, where
%the algorithm needs to check whether a $(2, 2)$-sequence exists. Now, a  $(2,
%2)$-sequence  exists:
%\begin{enumerate}[(a)]
%
%\item 	If there are (at least) four $ \mathrm{2}-cycles $ \cite{16}
%
%\item 	If there are two intersecting 2-cycles(Observation 13 Elias and Hartman
%\cite{36}
%
%\item 	 If there are two nonintersecting 2-cycles, apply a transposition on
%three of the four black edges of the two 2-cycles (check all four
%possibilities). This is a 2-move \cite{16}. There is a (2,2)-sequence iff, in
%the resulting graph, there is an oriented cycle.
%\item 	 Otherwise, the permutation is a 3-permutation. If all cycles are
%unoriented, there is no (2,2)-sequence.
%\item 	 Otherwise, for each oriented 3-cycle, check if, after applying a
%2-move on it, there is an oriented cycle in the resulting graph. There is a
%(2,2)-sequence iff the answer is yes for some cycle.
%
%
%\end{enumerate}
%\vspace{-.10in}
\begin{lemma}\label{lem-15}
Step \ref{Step2} of the EH algorithm can be implemented in $O(n\log n)$ time.
\end{lemma}

\begin{proof}
\vspace{-.10in}
To check whether a $(2, 2)$-sequence exists, the following sub-steps are executed:
%\vspace{-.15in}
\begin{enumerate}[(a)]

\item\label{ItemA} We check whether there are (at least) four ${2}$-cycles. If
yes, then we are done; otherwise we go to the next step.

\item\label{ItemB} If there are two intersecting 2-cycles then a $(2, 2)$-sequence
exists and we are done~\cite{36}. Otherwise we go to the following step.

\item\label{ItemC} If there are two nonintersecting 2-cycles, we apply a
transposition on three of the four black edges of the two 2-cycles (check all
four possibilities). Clearly, this is a 2-move~\cite{16}. Now, there is a
(2,2)-sequence iff in the resulting graph there is an oriented cycle.
Otherwise we go to the following step.

\item\label{ItemD} In this case the permutation is a 3-permutation. Here, if all
 cycles are unoriented, there is no (2,2)-sequence. Otherwise, for each oriented
3-cycle, we need to check if, after applying a 2-move on it, there is an
oriented cycle in the resulting graph. There is a (2,2)-sequence iff the answer
is yes for some cycle.
\end{enumerate}
%\vspace{-.15in}
Clearly, the complexity depends  on Sub-steps~\ref{ItemC} and~\ref{ItemD} as
these two cases involve applying the 2-move and the transpositions. Hence the result follows.\qed

%As stated in Lemma 13 ,transposition(as well as 2-move) can be implemented
%in$\mathrm{O}$(\textit{$\log$n}) time. So . Step 2 of  Sort can be implemented
%in $\mathrm{O}$(\textit{$n\log$n})) time.
\end{proof}

%\begin{lemma}\label{lem-16}
%\rm{\cite{31}} The number of even cycles in a breakpoint graph must be even.
%\end{lemma}
%\vspace{-.15in}
\begin{lemma}\label{lem-17}
Steps~\ref{Step3} and \ref{Step4} of  the EH algorithm can be implemented in $O(n\log n)$ time.
\end{lemma}
\begin{proof}
%\vspace{-.10in}
We have even number of 2-cycles in the breakpoint
graph for a simple permutation \cite{31}. A 2-move in Step \ref{Step3} transforms two
2-cycles into a 1-cycle and a 3-cycle. All the 2-cycles of G($\pi$) can be found
in linear time and be eliminated by at most $\frac{n}{2}$ 2-moves. Since one
transposition takes $O(\log n)$ time, Step~\ref{Step3} can be done in $O(n\log n)$ time. Finally, for Step~\ref{Step4}, all the 3-cycles can be marked by a linear scan of the breakpoint graph  which takes at most $O(n)$ time. Hence the result follows.\qed
%\vspace{-.05in}
\end{proof}

%
%\begin{lemma}\label{lem-20}
%Steps~\ref{Step5bStart} to~\ref{Step5bEnd} of the EH algorithm can be
%implemented in $O(\log n)$ time.
%\end{lemma}
%
%\begin{proof}There are two types of extensions that are sufficient for extending
%any cycle $\mathrm{C}$. These sufficient extensions are:
%\begin{enumerate}
%\item  [Type 1:] Extensions closing open gates, and
%\item  [Type 2:]Extensions of full configurations such that the extended
%configuration has at most one open gate.
%\end{enumerate}
%
%To do a sufficient extension of Type 1 (add a cycle that closes an open gate),
%we need to pick an arbitrary open gate and find another cycle that intersects
%with the open gate. For this, we query the permutation tree with the black edge
%$\langle b_i , b_j\rangle $ of the open gate under consideration. The query
%procedure in turn returns the intersecting pair $\langle  b_k,b_{l-1}\rangle  $
%as stated above. This step takes  $O(\log n)$ time.
%
%If the configuration is full, i.e., there are no open gates, we do sufficient
%extension of Type 2. To do this, we query the permutation tree with each pair of
%black edges of each cycle in the configuration, until we find a cycle that
%intersects with a pair. If such a cycle is found, we extend the configuration by
%this cycle to find a component of size greater than or equal to 9. As there can
%be atmost 24 such pairs of black edges, this step takes $O(\log n)$ time as well.
%
%Finally, we apply an $\frac{11}{8}$-sequence by using the transposition
%procedure of  permutation tree which takes $O(\log n)$ time. Hence the result
%follows.\qed
%\end{proof}
%\vspace{-.15in}
\begin{lemma}\label{lem-21}
The while loop at Step~\ref{Step5} of the EH algorithm can be implemented in
$O(n\log n)$ time.
\end{lemma}

\begin{proof} 
%\vspace{-.10in}
The loop iterates at most $n$ times and each iteration takes
$O(\log n)$ time as follows. Step~\ref{Step5a} runs in $O(\log n)$ time, because, to apply a 2-move, we use the transposition operation on the permutation tree. Now, consider Steps~\ref{Step5bStart} to~\ref{Step5bEnd}. There are two types of extensions that are sufficient for extending
any cycle $\mathrm{C}$:
%\vspace{-.15in}
\begin{enumerate}
\item  \emph{Type 1:} Extensions closing open gates, and
\item  \emph{Type 2:} Extensions of full configurations such that the extended
configuration has at most one open gate.
\end{enumerate}
%\vspace{-.15in}
To do a sufficient extension of Type 1 (add a cycle that closes an open gate),
we need to pick an arbitrary open gate and find another cycle that intersects
with the open gate. For this, we query the permutation tree with the black edge
$\langle b_i , b_j\rangle $ of the open gate under consideration. The query
procedure in turn returns the intersecting pair $\langle  b_k,b_{l-1}\rangle  $
as stated above. This step takes  $O(\log n)$ time.

If the configuration is full, i.e., there are no open gates, we do sufficient
extension of Type 2. To do this, we query the permutation tree with each pair of
black edges of each cycle in the configuration, until we find a cycle that
intersects with a pair. If such a cycle is found, we extend the configuration by
this cycle to find a component of size greater than or equal to 9. As there can
be atmost 24 such pairs of black edges, this step takes $O(\log n)$ time as well. Finally, we apply an $\frac{11}{8}$-sequence by using the transposition
procedure of  permutation tree which takes $O(\log n)$ time. Hence the result
follows.\qed
\end{proof}
%\vspace{-.10in}
\begin{lemma}\label{lem-22}
Steps~\ref{Step6} to ~\ref{Step8} of the EH algorithm can be implemented in $O(n\log n)$ time.
\end{lemma}
\begin{proof} 
%\vspace{-.10in}
In Step~\ref{Step6}, the application of an $\frac{11}{8}$-sequence takes $O(\log n)$ time and this step iterates at most $O(n)$ times. Now, applying a (3,2)-sequence is essentially equivalent to applying 3 transpositions such that at least 2 of them are 2-moves. Since, there can be no more than $n$ 3-cycles, Step~\ref{Step7} also runs in $O(n\log n)$ time. Hence, the result follows, since Step~\ref{Step8} of the EH algorithm can also be implemented in $O(n\log n)$ time~\cite{31}.\qed
\end{proof}
From the above lemmas, it is easy to see that the EH algorithm, implemented with permutation tree, runs in $O(n\log n)$ time.

%\section{Conclusion}\label{con}
%In this paper, we have used the permutation tree data structure of \cite{31} to
%improve the running time of the best known $\mathrm{1.375}-$approximation algorithm for
%sorting by transpositions \cite{34} from $O(n^2)$ to $O(n\log n)$. This
%improvement in running time is bound to have serious impact in the relevant
%research in computational biology especially due to huge amount of various
%genomic data (DNA, RNA, and protein sequences) becoming available for
%evolutionary distance measurement.
\bibliographystyle{abbrv}
\bibliography{reference}
\end{document}